\documentclass[runningheads]{llncs}

\usepackage{amssymb}
\usepackage{amsmath}
\usepackage[english]{babel}

\begin{document}

\title{Probabilistic Parameterized Polynomial Time\thanks{Research for this paper has been funded through NWO EW TOP grant 612.001.601.}}

\author{Nils Donselaar}

\authorrunning{N. Donselaar}

\institute{Radboud University, Donders Institute for Brain, Cognition and Behaviour\\ Montessorilaan 3, 6525 HR Nijmegen, The Netherlands}

\maketitle

\begin{abstract}
We examine a parameterized complexity class for randomized computation where only the error bound and not the full runtime is allowed to depend more than polynomially on the parameter, based on a proposal by Kwisthout in \cite{kwisthout1,kwisthout2}. We prove that this class, for which we propose the shorthand name \textsf{PPPT}, has a robust definition and is in fact equal to the intersection of the classes \textsf{paraBPP} and \textsf{PP}. This result is accompanied by a Cook-style proof of completeness for the corresponding promise class (under a suitable notion of reduction) for parameterized approximation versions of the inference problem in Bayesian networks, which is known to be \textsf{PP}-complete. With these definitions and results in place, we proceed by showing how it follows from this that derandomization is equivalent to efficient deterministic approximation methods for the inference problem. Furthermore, we observe as a straightforward application of a result due to Drucker in \cite{drucker} that these problems cannot have polynomial size randomized kernels unless the polynomial hierarchy collapses to the third level. We conclude by indicating potential avenues for further exploration and application of this framework.

\end{abstract}

\section{Preliminaries}\label{sec:intro}

The simple yet powerful idea which lies at the heart of the theory of parameterized complexity is that the hardness of computational problems may be better studied by analyzing the effects of particular aspects of its instances, treating these as a distinguished problem parameter and allowing the time (or other measures such as space) required to find a solution to depend on this parameter by an unbounded factor. This leads to an account of \textit{fixed-parameter tractability} (\textsf{FPT}) and a hardness theory based on classes belonging to the \textsf{W}-\textit{hierarchy} which together mirror the parts played by \textsf{P} and \textsf{NP} in classical complexity theory.\\
\\
In the two decades since the appearance of \cite{downeyfellows1}, the book by Downey and Fellows which largely formed the foundations of the field, research in parameterized complexity theory has gone far beyond this initial outlook and revealed a rich structure and many interesting questions to pursue, a lot of which covered in the updated \cite{downeyfellows2}. Yet with a few notable exceptions, little attention has been paid to probabilistic computation in the parameterized setting. The most encompassing effort thus far has been made by Montaya and M\"{u}ller in \cite{montoyamuller}, where they show amongst other things that the natural analogue \textsf{BPFPT} relates to other complexity classes in much the same way as does \textsf{BPP} in the classical setting.\footnote{We also mention \cite{chauhanrao} which studies \textsf{PFPT}, the parameterized counterpart to \textsf{PP}.}\\
\\
Our aim with the present paper is to improve on this situation by demonstrating that studying parameterized probabilistic computation amounts to more than simply reconstructing results from the classical setting (which can already be a non-trivial task, as evidenced by the work done in \cite{montoyamuller}), and that results obtained in this way can have broader theoretical and practical significance. In particular we study a complexity class intended to capture \textit{probabilistic parameterized polynomial time computability}, which we shall refer to as \textsf{PPPT} for this reason.\footnote{In \cite{kwisthout1,kwisthout2} this class was proposed under the name \textsf{FERT} for \textit{fixed-error randomized tractability}, intended to be reminiscent of \textsf{FPT}. We believe that the name \textsf{PPPT} is more appropriate as it calls into mind the class \textsf{PP} as well as ppt-reductions.}\\
\\
This class \textsf{PPPT} is informally defined by considering probabilistic algorithms for parameterized problems, except allowing not the runtime but instead only the error bound to depend on the parameter by more than a polynomial factor. As such, \textsf{PPPT} can be thought of as containing those problems in \textsf{PP} which are nevertheless close to being in \textsf{BPP} and hence randomized tractable in a certain sense. This perspective, which we will explore more rigorously later on, formed much of the motivation of the class's original proposal in \cite{kwisthout1}.\\
\\
In what follows, we assume the reader to be familiar with the basics of classical and parameterized complexity theory. However, we repeat the definitions of the complexity classes used here, mostly to facilitate the comparison with the class \textsf{PPPT} which we formally define in the next section. First recall the probabilistic complexity classes \textsf{BPP} and \textsf{PP}:

\begin{definition}\label{def:BPP}
\emph{\textsf{BPP}} is the class of decision problems computable in time \emph{$|x|^c$} for some constant $c$ by a probabilistic Turing machine which gives the correct answer with probability more than $\frac{1}{2} + |x|^{-d}$ for some constant $d$.
\end{definition}

\begin{definition}\label{def:PP}
\emph{\textsf{PP}} is the class of decision problems computable in time $|x|^c$ for some constant $c$ by a probabilistic Turing machine which gives the correct answer with probability more than $\frac{1}{2}$.
\end{definition}

We mostly follow the original definition presented in \cite{gill} in that we consider a probabilistic Turing machine $\mathcal{M}$ to be a Turing machine with access to random bits which it may query at every step of its execution, and whose transition function may depend on the values read off in this way. However, we include the generalization that a probabilistic Turing machine $\mathcal{M}$ may query not one but $r_{\mathcal{M}}$ many random bits at each step, where as usual we drop the subscript when it can be inferred from the context. Before continuing we note that $r_{\mathcal{M}} \leq \log |\mathcal{M}|$.

\begin{definition}\label{def:FPT}
\emph{\textsf{FPT}} is the class of parameterized decision problems computable in time $f(k)|x|^c$, where $f$ is a computable function in $k$ and $c$ is a constant.
\end{definition}

\begin{definition}\label{def:paraBPP}
\emph{\textsf{paraBPP}} is the class of parameterized decision problems computable in time $f(k)|x|^c$ by a probabilistic Turing machine which gives the correct answer with probability more than $\frac{1}{2} + |x|^{-d}$, where $f$ is a computable function in $k$ and $c$ and $d$ are constants.
\end{definition}

The method of converting classical complexity classes \textsf{C} to parameterized classes \textsf{paraC} illustrated above originates from \cite{flumgrohe}, and indeed $\mathsf{FPT} = \mathsf{paraP}$. One should keep in mind though that this construction does not yield the usual parameterized classes including \textsf{BPFPT}, as these are furthermore characterized by using at most $f(k)\log|x|$ random bits. (cf. \cite{chenflumgrohe2}). In fact, as observed in \cite{montoyamuller}, $\mathsf{BPFPT} = \mathsf{paraBPP}$ if and only if $\mathsf{P} = \mathsf{BPP}$. As we shall see in the next section, a similar statement remains true when we replace \textsf{paraBPP} by the class \textsf{PPPT}.

\section{Error Parameterization}\label{sec:PPPT}

We now provide a formal definition of the class \textsf{PPPT}:

\begin{definition}\label{def:PPPT}
We say that a parameterized decision problem $A$ is in \emph{\textsf{PPPT}} if there exist a computable function $f: \mathbb{N} \rightarrow
(0,\frac{1}{2}]$, a constant $c\in\mathbb{N}$ and a probabilistic Turing machine $\mathcal{M}$ which on input $(x,k)$ halts in time $(|x|+k)^c$ with probability at least $\frac{1}{2} + f(k)$ of giving the correct answer.
\end{definition}

Based on the convention that the parameter value $k$ is given as a unary string along with the rest of the input $x$, from this definition it is immediate that \textsf{PPPT} is a subclass of \textsf{PP}. Moreover, it can be shown that this definition is robust in two ways, which makes it easy to see that \textsf{PPPT} is also a subclass of \textsf{paraBPP}.

\begin{proposition}\label{prop:PPPT}
The class \textup{\textsf{PPPT}} in Definition~\ref{def:PPPT} remains the same if\\
(i) the error bound on a correct decision is $f(k)|x|^{-d}$ or $\textup{min}(f(k), |x|^{-d})$ instead; (ii) the probability
of a false positive is not bounded away from $\frac{1}{2}$.
\end{proposition}
\begin{proof}
For (i), note that $\text{min}(f(k)^2, |x|^{-2d}) \leq f(k)|x|^{-d} \leq \text{min}(f(k), |x|^{-d})$, which shows that such error bounds may be used interchangeably. For independent identically distributed Bernoulli variables with $p = \frac{1}{2}+\epsilon$, Hoeffding's inequality states that the probability of the average over $n$ trials being no greater than $\frac{1}{2}$ is at most $e^{-2n\epsilon^2}$. Thus for $\epsilon = \text{min}(f(k)^2, |x|^{-2d})$ we may take the majority vote over $n = |x|^{4d}$ runs to obtain a probability of at least $\frac{1}{2} + f(k)^2$ of giving the right answer. To see this, observe that the claim is trivially true whenever $\text{min}(f(k)^2, |x|^{-2d}) = f(k)^2$, while if $\text{min}(f(k)^2, |x|^{-2d}) = |x|^{-2d}$ then we should have $\frac{1}{2}-f(k)^2 \geq e^{-2}$, which is always the case since $f(k)^2 \leq \frac{1}{4}$ by definition. As this is only a polynomial number of repetitions, any error bound of either of these two alternative forms can be amplified to conform to the one required by Definition~\ref{def:PPPT} without exceeding the time restrictions.\\
\\
For (ii), we provide essentially the same argument as can be used to show that the class \textsf{PP} remains the same under this less strict requirement. Suppose that $\mathcal{M}$ is a probabilistic Turing machine which decides the problem $A$ in time $(|x|+k)^c$ for some $c$, with suitable $f$ and $d$ given such that \textit{Yes}-instances are accepted with probability at least $\frac{1}{2} + f(k)|x|^{-d}$ and \textit{No}-instances with probability at most $\frac{1}{2}$. Then we may consider a probabilistic Turing machine $\mathcal{M}'$ which on input $(x,k)$ operates the same up to where $\mathcal{M}$ would halt, at which point the outcome of $\mathcal{M}$ is processed as follows. In case of a rejection, the outcome is simply preserved, while an acceptance is rejected with probability $\frac{1}{2}f(k)|x|^{-d}$ (which is possible within the original polynomial time bound). Now a \textit{Yes}-instance $(x,k)$ of $\mathcal{M}$ is accepted by $\mathcal{M}'$ with probability at least $(\frac{1}{2} + f(k)|x|^{-d})(1-\frac{1}{2}f(k)|x|^{-d}) \geq \frac{1}{2} + \frac{1}{4}f(k)|x|^{-d}$, while \textit{No}-instances are accepted with probability at most $\frac{1}{2}-\frac{1}{4}f(k)|x|^{-d}$. Thus $\mathcal{M'}$ halts in time polynomial in $|x|+k$ and bounds the probability of a false positive away from $\frac{1}{2}$ by some $f(k)|x|^{-d}$ as required, which concludes the proof.
\end{proof}

One can extend the probabilistic amplification used in the proof of Proposition~\ref{prop:PPPT}(i) to (instead) remove the term $f(k)$ from the error bound, at the cost of introducing this factor into the runtime. This is generally not allowed for \textsf{PPPT}, but permissible for \textsf{paraBPP}, hence \textsf{PPPT} is a subclass of \textsf{paraBPP}. In fact, this inclusion is easily shown to be strict.

\begin{proposition}\label{prop:strict}
$\mathsf{PPPT} \subsetneq \mathsf{paraBPP}$, i.e.~$\mathsf{PPPT}$ is a strict subclass of $\mathsf{paraBPP}$.
\end{proposition}
\begin{proof}
Let $A$ be any decidable problem not in \textsf{PP}, and parameterize this problem by the input size, so that we obtain the problem
$|x|$-$A$. It is immediate that $|x|$-$A$ cannot be in \textsf{PPPT} since $A$ is not in \textsf{PP}. On the other hand, $|x|$-$A$ is
in $\mathsf{paraBPP}$ as the parameterization permits an arbitrary runtime for the decision procedure for $A$, hence the inclusion $\mathsf{PPPT}\subseteq \mathsf{paraBPP}$ is strict.
\end{proof}

Note that what the proof above actually shows is that there are problems in \textsf{FPT} which are not in \textsf{PPPT}. This tells us that
the inclusion in the statement remains strict even if $\mathsf{BPP} = \mathsf{P}$, in which case we have that $\mathsf{paraBPP} =
\mathsf{FPT}$: see Proposition~5.1 of \cite{montoyamuller}. Similar arguments can be used to show that $\mathsf{PPPT} \subseteq \mathsf{FPT}$ implies $\mathsf{BPP} = \mathsf{P}$. Most importantly however, we can extend these results by providing an exact characterization of \textsf{PPPT}, the proof of which relies on essentially the same strategy used to show that every problem in \textsf{FPT} is kernelizable.

\begin{theorem}\label{thm:PPPT}
$\mathsf{PPPT} = \mathsf{paraBPP}\cap\mathsf{PP}$. In particular, if a problem $k$-$A$ is in $\mathsf{paraBPP}$ and also in $\mathsf{PP}$ as an unparameterized problem, then $k$-$A$ is in $\mathsf{PPPT}$.
\end{theorem}
\begin{proof}
Suppose $k$-$A$ is as in the statement of the theorem. Then there is a probabilistic Turing machine $\mathcal{M}$ which on input $(x,k)$
halts in time $f(k)|x|^c$ and makes the correct decision with probability at least $\frac{1}{2} + |x|^{-d}$. Furthermore, there is a probabilistic Turing machine $\mathcal{M}'$ which on input $x$ halts in time $|x|^{c'}$ and makes the correct decision with probability at least $\frac{1}{2} +
2^{-r|x|^{c'}}$. Then on any given $(x,k)$, we can first run $\mathcal{M}$ for $|x|^{c+1}$ steps, adopting its decision if it halts within that time. If it does, then we have given the correct answer with probability at least $\frac{1}{2} + |x|^{-d}$. If it does not, then we may conclude that $f(k) > |x|$, in which case we proceed by running $\mathcal{M}'$ which will halt in time $|x|^{c'}$ with a probability of at least $\frac{1}{2}+2^{-r\cdot f(k)^{c'}} = \frac{1}{2}+ g(k)$ of giving the correct answer. Thus in time $\mathcal{O}(|x|^{c+c'+1})$ we can decide with probability at least $\frac{1}{2} + \text{min}(g(k),|x|^{-d})$ whether $(x,k)$ is in $k$-$A$, which means $k$-$A$ is in \textsf{PPPT} by Proposition~\ref{prop:PPPT}. As we already observed that \textsf{PPPT} is a subclass of both \textsf{paraBPP} and \textsf{PP}, this yields the conclusion that $\mathsf{PPPT} = \mathsf{paraBPP}\cap\mathsf{PP}$ as stated.
\end{proof}

For context it may be valuable to note that the idea of a complexity class being the intersection of a parameterized and a classical one has been considered before: while \cite{caichendowneyfellows} looked into $\mathsf{W[P]}\cap\mathsf{NP}$, \cite{montoyamuller} briefly discussed $\mathsf{BPFPT}\cap\mathsf{BPP}$. However, we believe it is important to note first of all that the class \textsf{PPPT} arose from a reasonably natural definition intended to capture a slightly weaker kind of parameterization, instead of its definition being explicitly constructed to ensure a correspondence to the intersection of \textsf{paraBPP} and \textsf{PP}. Furthermore, and perhaps most importantly, we can actually exhibit natural problems which are complete for (the promise version of) the class \textsf{PPPT}, something which has not yet been done for \textsf{BPP}. In the remainder of this section we shall describe these problems, which originate from the domain of approximate Bayesian inference. We subsequently prove completeness for these problems in Section~\ref{sec:redux}, which sets us up to explore the main implications of our results in Section~\ref{sec:applic}.\\
\\
Below we provide a definition of a Bayesian network, so that we may introduce the problem of inference within such networks; for the reader interested in a more detailed treatment, we refer to a standard textbook such as \cite{kollerfriedman}.

\begin{definition}\label{def:BN}
A \textit{Bayesian network} is a pair $\mathcal{B} = (G, \textup{Pr})$, where $G = (V,A)$ is a directed acyclic graph whose nodes represent statistical variables, and $\textup{Pr}$ is a set of families of probability distributions containing for each node $V$, and each possible configuration $c_{\rho(V)}$ of the variables represented by its parents, a distribution $\textup{Pr}(c_V \mid c_{\rho(V)})$ over the possible outcomes of its represented variable.
\end{definition}

As reflected by the notation, one typically blurs the distinction between the node $V$ and the statistical variable which it represents, so that one may use $\mathrm{\Omega}(V)$ to refer directly to the set of possible outcomes of this variable.\\
\\
One of the main computational problems associated with Bayesian networks is that of \textit{inference}, which is to determine what the likelihood is of some given combination of outcomes, possibly conditioned on certain specified outcomes for another set of variables. Below we describe the corresponding decision problem.\\
\\
\textsc{Bayesian Inference}\\
\textit{Input:} A Bayesian network $\mathcal{B} = (G, \text{Pr})$, two sets of variables $H,E\subseteq V(G)$, joint value assignments $h\in\mathrm{\Omega}(H)$ and $e\in\mathrm{\Omega}(E)$, a rational threshold value $0\leq q\leq 1$.\\
\textit{Question:} Is $\text{Pr}(h\mid e) > q$?\\
\\
Based on whether $E=\emptyset$, we shall refer to this problem as \textsc{Inference} or \textsc{Conditional Inference} respectively, both of which are \textsf{PP}-complete (see also \cite{parkdarwiche}). Despite their equivalence from a classical perspective, in terms of parameterized approximability the latter is more difficult: \cite{kwisthout2} has an overview of the main results known thus far. We consider the following specific parameterizations:\\
\\
$\epsilon$-\textsc{Inference}\\
\textit{Input:} A Bayesian network $\mathcal{B} = (G, \text{Pr})$, a set of variables $H\subseteq V(G)$, a joint value assignment $h\in\mathrm{\Omega}(H)$, and rational values $0\leq q\leq 1$ and $0 < \epsilon \leq \frac{1}{2}$.\\
\textit{Parameter:} $\lceil - \log \epsilon \rceil$.\\
\textit{Promise:} $\text{Pr}(h)\not\in(q-\epsilon,q+\epsilon)$.\\
\textit{Question:} Is $\text{Pr}(h) > q$?\\
\\
$\{\epsilon, \text{Pr}(h),\text{Pr}(e)\}$-\textsc{Conditional Inference}\\
\textit{Input:} A Bayesian network $\mathcal{B} = (G, \text{Pr})$, two sets of variables $H,E\subseteq V(G)$, joint value assignments $h\in\mathrm{\Omega}(H)$, $e\in\mathrm{\Omega}(E)$, rational values $0\leq q\leq 1$, $\epsilon>0$.\\
\textit{Parameters:} $\epsilon$, $\text{Pr}(h)$, and $\text{Pr}(e)$.\\
\textit{Promise:} $\text{Pr}(h)\not\in(q(1+\epsilon)^{-1},q(1+\epsilon))$.\\
\textit{Question:} Is $\text{Pr}(h\mid e) > q$?\\
\\
In what follows we use $\text{Pr}(h)\pm\epsilon \geq q$ as a shorthand for $\text{Pr}(h)>q$ with the promise that $\text{Pr}(h)\not\in(q-\epsilon,q+\epsilon)$, and similarly for the relative approximation; for more on this approach, see \cite{marx}. For completeness' sake, we remind the reader of the definition of a promise problem, based on \cite{goldreich}.

\begin{definition}\label{def:propro}
A promise problem $A$ consists of disjoint sets $A_{\textit{Yes}}$ and $A_{\textit{No}}$ of \textit{Yes} and \textit{No}-instances respectively, where $A_{\textit{Yes}}\cup A_{\textit{No}}$ may be a strict subset of the set of all inputs; $A_{\textit{Yes}}\cup A_{\textit{No}}$ is called the promise of the problem $A$.
\end{definition}

We can separate promise problems into complexity classes just as we do for decision problems by converting the familiar definitions in the following way.

\begin{definition}\label{def:proclass}
Let $C$ be any complexity class of decision problems. The class $pC$ of promise problems is defined by applying $C$'s criteria of membership to promise problems instead, e.g.~$pP$ is the class of promise problems $A$ for which there exists a polynomial-time algorithm which answers correctly on its promise.
\end{definition}

We now show for both of the approximate inference problems given above that they are in the promise version of \textsf{PPPT}, the former by an explicit argument, the latter by an appeal to Theorem~\ref{thm:PPPT}.

\begin{proposition}\label{prop:epinf}
$\epsilon$-\textsc{Inference} is in \textup{\textsf{pPPPT}}.
\end{proposition}
\begin{proof}
By forward sampling the network (i.e.~generating outcomes according to the distributions of each variable, following some topological ordering of the graph) and accepting with probability $1-\frac{q}{2}$ if the sample agrees with $h$, and with probability $\frac{1}{2}-\frac{q}{2}$ if it does not, we arrive at a probability of acceptance of $\frac{1}{2} + \frac{1}{2}(\text{Pr}(h)-q)$. Under the promise that $\text{Pr}(h)\not\in(q-\epsilon,q+\epsilon)$, the probability of giving the correct answer is now at least $\frac{1}{2} + \frac{\epsilon}{2}$, hence the problem is in \textsf{pPPPT}.
\end{proof}

\begin{proposition}\label{prop:ephecinf}
$\{\epsilon, \textup{Pr}(h),\textup{Pr}(e)\}$-\textsc{Conditional Inference} is in \textup{\textsf{pPPPT}}.
\end{proposition}
\begin{proof}
In \cite{henrion} it is shown that rejection sampling (which is forward sampling and dismissing the outcome if it does not agree with $e$) provides an algorithm which places $\{\epsilon, \text{Pr}(h),\text{Pr}(e)\}$-\textsc{Conditional Inference} in \textsf{paraBPP}.\footnote{Note that the parameter \text{Pr}(h) is only necessary here because we ask for a relative approximation: the same holds when considering \textsc{Inference} instead.} Because \textsc{Conditional Inference} is itself in \textsf{PP}, by Theorem~\ref{thm:PPPT} we conclude that the parameterized problem is in \textsf{pPPPT}.
\end{proof}

\section{Reductions and Completeness}\label{sec:redux}

At this point we wish to show that the two parameterized problems considered in the previous section are actually complete for the class \textsf{pPPPT}. In order to do this, we first determine which notion of reduction is the most suitable with respect to the class \textsf{PPPT}, after which we identify the machine acceptance problem for \textsf{pPPPT} and demonstrate its completeness for the class under these reductions. We then construct an explicit reduction from this problem to $\epsilon$-\textsc{Inference}, and in turn reduce the latter to $\{\epsilon, \text{Pr}(h),\text{Pr}(e)\}$-\textsc{Conditional Inference}, thereby establishing completeness for both of these problems.\\
\\
First of all, it is evident that while some form of parameterized reduction is required, the usual fpt-reductions are unsuitable because they allow the runtime to contain a factor superpolynomial in the parameter value and so \textsf{PPPT} is not closed under these. Furthermore, the reductions cannot be probabilistic either: while this is possible for \textsf{BPP} since the error can be reduced to constant using probabilistic amplification, mitigating the parameterized error bound is generally impossible without parameterized runtime (unless $\mathsf{P} = \mathsf{PP}$). Thus in this context it makes sense to consider the notion of a \textit{ppt-reduction}\footnote{As \cite{downeyfellows2} observes, the acronym `ppt' can be read equally well as either ``polynomial parameter transformation'' or ``parameterized polynomial transformation''.}, which was formally introduced in \cite{bodlaenderthomasseyeo}:

\begin{definition}\label{def:ppt}
A \emph{ppt-reduction} from $A$ to $B$ is a polynomial-time computable function $h: (x,k) \mapsto (x',k')$ such that there exists a polynomial $g$ with the property that $k' \leq g(k)$, and furthermore $(x,k)\in A$ if and only if $(x',k')\in B$. 
\end{definition}

However, we can remove the constraint that $k'$ is bounded by a polynomial in $k$, and instead simply demand that its value depends only on $k$.\footnote{I hereby express my gratitude to the anonymous reviewer who raised this point.} The resulting class of reductions is a slightly broader one for which we introduce the name \textit{pppt-reduction}, for which it is easy to see that \textsf{PPPT} is closed under pppt-reductions. Thus we would like to exhibit a parameterized problem which is complete for \textsf{PPPT} under pppt-reductions; yet here we run into the same issue as with the class \textsf{BPP}, namely that it may be impossible to effectively decide whether a given probabilistic Turing machine has a suitably lower-bounded probability to be correct on all possible inputs. Hence we have to add this explicit requirement in the form of a promise, which means we arrive at the machine acceptance problem stated below, and instead study completeness for the promise class \textsf{pPPPT}.\\
\\
\textsc{Error PTM Acceptance}\\
\textit{Input:} Two strings $x$ and $1^n$ and a probabilistic Turing machine $\mathcal{M}$.\\
\textit{Parameter:} A positive integer $k$.\\
\textit{Promise:} For any valid input to $\mathcal{M}$ and after any number of steps, the probability of acceptance does not lie between $\frac{1}{2} - 2^{-k}$ and $\frac{1}{2} + 2^{-k}$.\\
\textit{Question:} After $n$ steps, does $\mathcal{M}$ accept $x$ more often than it rejects?

\begin{proposition}\label{prop:pEPTMA}
\textsc{Error PTM Acceptance} is complete for \textup{\textsf{pPPPT}}.
\end{proposition}
\begin{proof}
The problem \textsc{Error PTM Acceptance} is straightforwardly seen to lie in \textsf{pPPPT}, as it can be decided by running $\mathcal{M}$ for $n$ steps on input $x$ and accepting only if it halts in an accepting state. Based on the promise this has a probability of at least $\frac{1}{2}+2^{-k}$ of making the correct decision, while taking time polynomial in the input size, which satisfies the requirements for being in \textsf{PPPT}.\\
\\
In turn, presenting a reduction from a problem $A\in\mathsf{pPPPT}$ to \textsc{Error PTM Acceptance} may be done in the following way. Suppose that $f$, $c$ and $\mathcal{M}$ witness that $A\in\mathsf{pPPPT}$, i.e.~on input $(x,k)$ the machine $\mathcal{M}$ will run in time at most $(|x|+k)^c$ and accept or reject based on whether $(x,k)\in A$ with probability at least $\frac{1}{2}+f(k)$. Then we can construct $\mathcal{M}'$ in polynomial time which on input $(x,k)$ simulates the machine $\mathcal{M}$ for exactly $(|x|+k)^c$ steps, deferring the decision until then: this will itself take time at most $(|x|+k)^a$ for some constant $a\geq c$.\\
\\
Using the above, we find that $(x,k)\in A$ precisely when after $|(x,k)|^a$ steps $\mathcal{M}'$ accepts $(x,k)$ more often than it rejects, hence $(x,k)\in A$ if and only if $((x,k),1^{(|x|+k)^a},\mathcal{M}',\lceil-\log f(k)\rceil)\in\textsc{Error PTM Acceptance}$. Since the remaining data (i.e.~other than the machine $\mathcal{M}'$) can also be given in polynomial time, this describes a pppt-reduction as required to show that \textsc{Error PTM Acceptance} is complete for \textsf{pPPPT} under pppt-reductions.
\end{proof}

We can now show the problem $\epsilon$-\textsc{Inference} to be complete for \textsf{pPPPT} as well, by providing what is essentially a Cook-style construction of a Bayesian network from the probabilistic Turing machine specification and number of steps and the input which together make up an instance of \textsc{Error PTM Acceptance}. In contrast to the previous result, the reduction resulting from this construction is a ppt-reduction rather than a pppt-reduction. 

\begin{theorem}\label{thm:einfcomp}
$\epsilon$-\textsc{Inference} is complete for \textup{\textsf{pPPPT}}.
\end{theorem}
\begin{proof}
Given an instance $(x,1^n,\mathcal{M},k)$ of \textsc{Error PTM Acceptance} we describe a ppt-reduction to $\epsilon$-\textsc{Inference} as follows. First, we construct the underlying graph of the Bayesian network $\mathcal{B}$ by stacking $n+1$ layers of nodes and connecting these using an intermediate gadget. Any such layer $i$ consists of $n+1$ nodes $X_{i,0},\ldots,X_{i,n}$ representing the potentially reachable cells of the machine tape, a pair of nodes $TH_i$ and $MS_i$ which track the current tape head position and machine state respectively, and a series of $r$ nodes $B_{i,1},\ldots,B_{i,r}$ which act as the random bits which the machine uses to determine its next step. This means that $\mathrm{\Omega}(X_{i,j})$ consists of the tape alphabet (including blanks), $\mathrm{\Omega}(TH_i) = \{0,\ldots,n\}$, $\mathrm{\Omega}(MS_i)$ is the set of machine states, and finally $\mathrm{\Omega}(B_{i,j}) = \{0,1\}$.\\
\\
Such a layer of nodes $i$ is connected to its successor through a gadget consisting again of $n+1$ nodes $Y_{i,0},\ldots,Y_{i,n}$, with the parents $\rho(Y_{i,0})$ of $Y_{i,0}$ being $TH_i$ and $X_{i,0}$ and $\rho(Y_{i,j+1}) = \{Y_{i,j},X_{i,j+1}\}$. These nodes $Y_{i,j}$ can be interpreted as storing the position of the tape head at step $i$ and reading off the tape until the correct cell $X_{i,j}$ is encountered, after which its symbol is copied and carried over all the way to $Y_{i,n}$. To achieve this, we require $\mathrm{\Omega}(Y_{i,j})$ to be the disjoint union of $\mathrm{\Omega}(TH_i)$ and $\mathrm{\Omega}(X_{i,j})$. Now the layer $i$ combined with its gadget is connected to the next one by setting $\rho(X_{i+1,j}) = \{X_{i,j},MS_i,TH_i,Y_{i,n},B_{i,1},\ldots,B_{i,n}\}$, $\rho(TH_{i+1}) = \{MS_i,TH_i,Y_{i,n},B_{i,1},\ldots,B_{i,n}\}$ and $\rho(MS_{i+1}) = \{MS_i,Y_{i,n},B_{i,1},\ldots,B_{i,n}\}$.\\
\\
We now have to assign probability distributions to each of these nodes such that they fulfill their intended purposes. First of all, the nodes $B_{i,j}$ are all uniformly distributed so that they may be correctly regarded as random bits. As for the first row, the remaining nodes are fixed to the first $n+1$ cells of the tape input, the tape head starting location and the initial state of the machine. All other nodes in the network have similar distributions which are deterministic given the values of their parents. In particular, $Y_{i,j} = X_{i,j}$ if $Y_{i,j-1} = j$ and $Y_{i,j} = Y_{i,j-1}$ otherwise (here $TH_i$ should be read for $Y_{i,-1}$), $X_{i+1,j} = X_{i,j}$ unless $TH_i = j$ in which case $MS_i, Y_{i,n}$ and $B_{i,1},\ldots,B_{i,n}$ together determine the symbol overwriting the previous one according to the transition function of $\mathcal{M}$, and in general the values of $TH_{i+1}$ and $MS_{i+1}$ follow from those of its parents based on this transition function as well.\\
\\
The reduction can now be straightforwardly expressed as follows: an instance $(x,1^n,\mathcal{M},k)$ is mapped to an instance $(\mathcal{B},MS_n,s_{accept},\frac{1}{2},k)$ of $\epsilon$-\textsc{Inference}, where $\mathcal{B}$ is constructed as above and $s_{accept}$ is the accepting state of $\mathcal{M}$. Then as required we have that after $n$ steps $\mathcal{M}$ accepts $x$ more often than it rejects if and only if $\text{Pr}(MS_n= s_{accept})\pm 2^{-k} \geq \frac{1}{2}$. Since $\mathcal{B}$ is of size polynomial in $n$ and $|\mathcal{M}|$ (in particular because the conditional probability distribution at every node is of polynomial size) and the parameter remains unchanged, this indeed describes a ppt-reduction, which completes the proof.
\end{proof}

In turn, we can reduce $\epsilon$-\textsc{Inference} to $\{\epsilon,\text{Pr}(h),\text{Pr}(e)\}$-\textsc{Conditional Inference}, thereby extending the hardness and hence completeness to the latter.

\begin{corollary}\label{cor:eheinfcomp}
$\{\epsilon,\text{Pr}(h),\text{Pr}(e)\}$-\textsc{Conditional Inference} is $\mathsf{pPPPT}$-complete.
\end{corollary}
\begin{proof}
Given an instance $(\mathcal{B},H,h,q,\epsilon)$ of $\epsilon$-\textsc{Inference}, we can adjust $\mathcal{B}$ by building an inverse binary tree below the nodes in $H$, with terminal node $T_H$ being $h$ when $H=h$ and $\neg h$ otherwise. We then furthermore add an initial, uniformly distributed binary node $R$ and another binary node $S$ with parents $R$ and $T_H$, distributed as follows:
\begin{equation*}
  \text{Pr}(s \mid R, T_H) = \left \{
  \begin{aligned}
    &1 && \text{for}\ R = r, T_H = h\\
    &0 && \text{for}\ R = \neg r, T_H = h\\
    &\frac{1}{2} && \text{otherwise}
  \end{aligned} \right.
\end{equation*}
Now $\text{Pr}(r) = \text{Pr}(s) = \frac{1}{2}$, and moreover $\text{Pr}(r\mid s) = \frac{1}{2} + \frac{1}{2}\text{Pr}(h)$, hence we find that $\text{Pr}(r\mid s) \pm \frac{1}{2}\epsilon \geq \frac{1}{2} + \frac{1}{2}q$ if and only if $\text{Pr}(h)\pm\epsilon\geq q$. This therefore describes a pppt-reduction from $\epsilon$-\textsc{Inference} to $\{\epsilon,\text{Pr}(h),\text{Pr}(e)\}$-\textsc{Conditional Inference}, albeit not a ppt-reduction as an artefact of the particular choice of parameter value corresponding to $\epsilon$. The result then follows from Theorem~\ref{thm:einfcomp}.
\end{proof}

To conclude this section, we discuss a question which may have occurred to the reader, namely whether one could simplify this approach by avoiding the inference problems altogether and working instead with the following variant of \textsc{MajSat}, which is the satisfiability problem complete for \textsf{PP}.\\
\\
\textsc{Gap-MajSat}\\
\textit{Input:} A propositional formula $\varphi$.\\
\textit{Parameter:} A positive integer $k$.\\
\textit{Promise:} The ratio of satisfying truth assignments of $\varphi$ does not lie between $\frac{1}{2}-2^{-k}$ and $\frac{1}{2}+2^{-k}$.\\
\textit{Question:} Is $\varphi$ satisfied by more than half of its possible truth assignments?\\
\\
The issue here is that the canonical reduction from \textsc{Error PTM Acceptance} to \textsc{Gap-MajSat} requires a number of variables proportional to both $n$ and the size of the machine $\mathcal{M}$, hence the original margin of $\pm 2^{-k}$ will shrink by a factor in the input size. The resulting parameter for the \textsc{Gap-MajSat} instance will thus depend on the input size, which means this reduction is not even an fpt-reduction. This points to a phenomenon also observed in the \textsf{W}-hierarchy, where \textsf{W[SAT]} (which is defined in terms of a parameterized satisfiability problem) is believed to a proper subclass of \textsf{W[P]} (which is defined in terms of a parameterized circuit satisfiability or machine acceptance problem). That the reduction does work for the inference problems suggests that Bayesian networks do have the direct expressive power of Turing machines lacked by propositional formulas.

\section{Application of Results}\label{sec:applic}

Ultimately, one of the main open questions in the area of probabilistic computation is whether $\mathsf{P} = \mathsf{BPP}$. In contrast to the more famous open question whether $\mathsf{P} = \mathsf{NP}$, the generally accepted view is that \textsf{BPP} is likely to equal \textsf{P}, based on works such as \cite{impagliazzowigderson}. However, due to the lack of natural problems which are known to be \textsf{BPP}-complete, it has not been possible to focus efforts on proving a particular problem to lie in \textsf{P} in order to demonstrate the collapse of \textsf{BPP}. We believe that our work makes an important contribution in that it indirectly provides a problem which can play this part, namely $\epsilon$-\textsc{Inference}. This relies in part on the following proposition adapted from \cite{montoyamuller} which we hinted at earlier.

\begin{proposition}\label{prop:PPPTequiv}
$\mathsf{PPPT} \subseteq \mathsf{FPT}$ if and only if $\mathsf{P} = \mathsf{BPP}$.
\end{proposition} 
\begin{proof}
Suppose $\mathsf{PPPT}\subseteq\mathsf{FPT}$, and let $A$ be an arbitrary problem in $\mathsf{BPP}$. Then certainly $A\in\mathsf{PP}$, and also $A\in\mathsf{paraBPP}$ for any constant parameterization, hence $A\in\mathsf{PPPT}$ by Theorem~\ref{thm:PPPT}. By assumption it follows that $A\in\mathsf{FPT}$, hence there is a deterministic algorithm for $A$ which runs in time $f(k)|x|^c$. But now the factor $f(k)$ is a constant term, which means $A$ is actually in \textsf{P} by this algorithm.\\
Conversely, suppose $\mathsf{P} = \mathsf{BPP}$, and let $A$ be an arbitrary problem in $\mathsf{PPPT}$ with corresponding error bound function $f(k)$. Given an instance $(x,k)$ of $A$ we can determine whether $|x| \leq f(k)^{-1}$: for the instances where this is true, the problem is in \textsf{FPT}, while it is in \textsf{BPP} for those where it is false. By assumption the latter problem is moreover in \textsf{P}, which means the entire problem $A$ is in \textsf{FPT}.
\end{proof}

Combined with Theorem~\ref{thm:einfcomp}, we arrive at the following result:

\begin{theorem}\label{thm:approx}
$\mathsf{P} = \mathsf{BPP}$ if and only if there exists an efficient deterministic absolute approximation algorithm for \textsc{Inference}, i.e.~a deterministic approximation which runs in time $f(\epsilon^{-1})|x|^c$ for some constant $c$ and computable function $f$.
\end{theorem}
\begin{proof}
Since $\epsilon$-\textsc{Inference} $\in\mathsf{paraBPP}$ and $\mathsf{paraBPP} = \mathsf{FPT}$ whenever $\mathsf{P} = \mathsf{BPP}$, this part of the equivalence is already established. By Theorem~\ref{thm:einfcomp} we know that $\epsilon$-\textsc{Inference} is \textsf{pPPPT}-complete, hence $\mathsf{PPPT}\subseteq\mathsf{FPT}$ if the problem has an fpt-algorithm, by which it follows from Proposition~\ref{prop:PPPTequiv} that $\mathsf{P} = \mathsf{BPP}$.
\end{proof}

At the same time, we can provide some indication as to the hardness of $\epsilon$-\textsc{Inference} by means of the framework of kernelization lower bounds, which is where the notion of a ppt-reduction originated. Here we consider the following reformulation, inspired by \cite{dell}, of a theorem by Drucker found in \cite{drucker}.

\begin{theorem}\label{thm:drucker}
If $A$ is an \textsf{NP}-hard or \textsf{coNP}-hard problem and $B$ is a parameterized problem such that there exists a polynomial-time algorithm which maps any tuple $(x_1,\ldots,x_t)$ of $n$-sized instances of $A$ to an instance $y$ of $B$ such that
\begin{enumerate}
\item if all $x_i$ are \textit{No}-instances of $A$, then $y$ is a \textit{No}-instance of $B$;
\item if exactly one $x_i$ is a \textit{Yes}-instance of $A$, then $y$ is a \textit{Yes}-instance of $B$;
\item the parameter $k$ of $y$ is bounded by $t^{o(1)}n^c$ for some constant $c$;
\end{enumerate}
then $B$ has no randomized (two-sided constant error) polynomial-sized kernels unless $\mathsf{coNP}\subseteq\mathsf{NP}/\mathsf{poly}$, collapsing the polynomial hierarchy to the third level.
\end{theorem}

We can use this Theorem to prove that $\epsilon$-\textsc{Inference} has no randomized polynomial-sized kernels unless the polynomial hierarchy collapses.

\begin{proposition}\label{prop:kernel}
$\epsilon$-\textsc{Inference} has no randomized polynomial-sized kernel unless $\mathsf{coNP}\subseteq\mathsf{NP}/\mathsf{poly}$. 
\end{proposition}
\begin{proof}
Consider the \textsf{NP}-hard problem \textsc{SAT}, and let $\varphi_1,\ldots,\varphi_t$ be propositional formulas in $n$ variables. We can rename the variables so that every formula uses the same $x_1,\ldots,x_n$ if necessary, introduce a new variable $x_0$, and take the disjunction $\psi = \bigvee^t_{i=1}\varphi_i \vee x_0$. Then $\psi$ is a formula with $n+1$ variables with a majority of its truth assignments being satisfying if and only if at least one of the $\varphi_i$ is satisfiable, hence $(\psi,n+1)\in\textsc{Gap-MajSat}$ if and only if $\varphi_i\in\textsc{SAT}$ for some $i$. Thus by Theorem~\ref{thm:drucker} \textsc{Gap-MajSat} does not have randomized polynomial-sized kernels unless $\mathsf{coNP}\subseteq\mathsf{NP}/\mathsf{poly}$. Furthermore, by \cite{bodlaenderthomasseyeo} this property is closed under ppt-reductions, and the usual reduction from \textsc{Gap-MajSat} to $\epsilon$-\textsc{Inference} (which amounts to constructing a Boolean circuit out of the given formula) is in fact a ppt-reduction, hence neither does $\epsilon$-\textsc{Inference} have randomized polynomial-sized kernels under the assumption that $\mathsf{coNP}\not\subseteq\mathsf{NP}/\mathsf{poly}$.
\end{proof}

While perhaps unsurprising, this result serves in particular as a reminder that hard problems in \textsf{PPPT} such as $\epsilon$-\textsc{Inference} are not solvable by polynomial kernelization followed by a probabilistic (\textsf{PP}) algorithm.

\section{Closing Remarks}

In this paper we have explored the proposal made in \cite{kwisthout1,kwisthout2} of an alternative parameterized randomized complexity class, which we have called \textsf{PPPT} and of which we have shown that it is identical to the intersection of $\textsf{PP}$ and $\textsf{paraBPP}$. In the preceding sections we showed that the problem $\epsilon$-\textsc{Inference} is a natural fit for this class, as it is not only a member of the class in a straightforward way (Proposition~\ref{prop:epinf}), it is moreover complete for the corresponding problem class (Theorem~\ref{thm:einfcomp}). Because of the close relation between classical and parameterized probabilistic computation (Proposition~\ref{prop:PPPTequiv}), the class \textsf{PPPT} turns out to have unexpected broader relevance, as finding an efficient deterministic absolute approximation algorithm for \textsc{Inference} is necessary and sufficient for the derandomization of \textsf{BPP} to \textsf{P} (Theorem~\ref{thm:approx}).\\
\\
In other words, we are in the fortunate circumstances where efforts to address a long-standing open question originating in theory can actually coincide with the search for a novel algorithm capable of solving a practical problem, and most importantly the existence of such an algorithm actually follows from a conjecture supported by other considerations. With this paper we wish to call attention to this opportunity for researchers with theoretical and practical motivations alike to engage with a challenge which is broadly relevant to multiple research communities at once. It is our hope that such focused efforts on the $\epsilon$-\textsc{Inference} problem may lead to a valuable breakthrough in both the fields of structural complexity theory and of probabilistic graphical models.

\subsection*{\ackname}
The author thanks Johan Kwisthout and Hans Bodlaender for sharing insightful remarks in his discussions with them, and also Ralph Bottesch for providing useful comments on an early draft of this paper.

\bibliographystyle{./splncs04}
\bibliography{references}

\begin{thebibliography}{10}
\providecommand{\url}[1]{\texttt{#1}}
\providecommand{\urlprefix}{URL }
\providecommand{\doi}[1]{https://doi.org/#1}

\bibitem{bodlaenderthomasseyeo}
Bodlaender, H.L., Thomass\'{e}, S., Yeo, A.: Kernel bounds for disjoint cycles
  and disjoint paths. Theoretical Computer Science  \textbf{412},  4570--4578
  (2011)

\bibitem{caichendowneyfellows}
Cai, L., Chen, J., Downey, R.G., Fellows, M.R.: On the structure of
  parameterized problem in {NP}. In: Enjalbert, P., Mayr, E.W., Wagner, K.W.
  (eds.) Proceedings of STACS 94. pp. 507--520 (1994)

\bibitem{chauhanrao}
Chauhan, A., Rao, B.V.R.: Parameterized analogues of probabilistic computation.
  In: Ganguly, S., Krishnamurti, R. (eds.) Algorithms and Discrete Applied
  Mathematics. pp. 181--192 (2015)

\bibitem{chenflumgrohe2}
Chen, Y., Flum, J., Grohe, M.: Machine-based methods in parameterized
  complexity theory. Theoretical Computer Sciences  \textbf{339},  167--199
  (2005)

\bibitem{dell}
Dell, H.: {AND}-compression of {NP}-complete problems: Streamlined proof and
  minor observations. Algorithmica  \textbf{75},  403--423 (2016)

\bibitem{downeyfellows1}
Downey, R.G., Fellows, M.R.: Parameterized Complexity. Springer (1999)

\bibitem{downeyfellows2}
Downey, R.G., Fellows, M.R.: Fundamentals of parameterized complexity. Springer
  (2013)

\bibitem{drucker}
Drucker, A.: New limits to classical and quantum instance compression. Tech.
  Rep. TR12-112, Electronic Colloquium on Computational Complexity (ECCC)
  (2014), \url{http://eccc.hpi-web.de/report/2012/112/}

\bibitem{flumgrohe}
Flum, J., Grohe, M.: Describing parameterized complexity classes. Information
  and Computation  \textbf{187},  291--319 (2003)

\bibitem{gill}
Gill, J.: Computational complexity of probabilistic {Turing} machines. SIAM
  Journal on Computing  \textbf{6}(4),  675--695 (1977)

\bibitem{goldreich}
Goldreich, O.: On promise problems: A survey. In: Goldreich, O., Rosenberg,
  A.L., Selman, A.L. (eds.) Theoretical Computer Science: Essays in Memory of
  Shimon Even, pp. 254--290. Springer (2006)

\bibitem{henrion}
Henrion, M.: Propagating uncertainty in {Bayesian} networks by probabilistic
  logic sampling. In: Lemmer, J.F., Kanal, L.N. (eds.) Uncertainty in
  Artificial Intelligence, Machine Intelligence and Pattern Recognition,
  vol.~5, pp. 149--163 (1988)

\bibitem{impagliazzowigderson}
Impagliazzo, R., Wigderson, A.: P = {BPP} if {E} requires exponential circuits:
  Derandomizing the {XOR} lemma. In: Proceedings of STOC '97. pp. 220--229
  (1997)

\bibitem{kollerfriedman}
Koller, D., Friedman, N.: Probabilistic graphical models: principles and
  techniques. MIT Press (2009)

\bibitem{kwisthout1}
Kwisthout, J.: Tree-width and the computational complexity of {MAP}
  approximations in {Bayesian} networks. Journal of Artificial Intelligence
  Research  \textbf{53},  699--720 (2015)

\bibitem{kwisthout2}
Kwisthout, J.: Approximate inference in {Bayesian} networks: parameterized
  complexity results. International Journal of Approximate Reasoning
  \textbf{93},  119--131 (2018)

\bibitem{marx}
Marx, D.: Parameterized complexity and approximation algorithms. The Computer
  Journal  \textbf{51},  60--78 (2008)

\bibitem{montoyamuller}
Montoya, J.A., M\"{u}ller, M.: Parameterized random complexity. Theory of
  Computing Systems  \textbf{52},  221--270 (2013)

\bibitem{parkdarwiche}
Park, J.D., Darwiche, A.: Complexity results and approximation strategies for
  {MAP} explanations. Journal of Artificial Intelligence Research  \textbf{21},
   101--133 (2004)

\end{thebibliography}

\end{document}